\documentclass[english,submission,copyright,creativecommons]{eptcs}
\usepackage[T1]{fontenc}
\usepackage[latin9]{inputenc}
\usepackage{amsthm}
\usepackage{amsmath}

\makeatletter

\providecommand{\tabularnewline}{\\}

\theoremstyle{plain}
\newtheorem{thm}{\protect\theoremname}
\theoremstyle{definition}
\newtheorem{example}[thm]{\protect\examplename}
\theoremstyle{plain}
\newtheorem{prop}[thm]{\protect\propositionname}
\theoremstyle{plain}
\newtheorem{lem}[thm]{\protect\lemmaname}
\ifx\proof\undefined\
\newenvironment{proof}[1][\protect\proofname]{\par
\normalfont\topsep6\p@\@plus6\p@\relax
\trivlist
\itemindent\parindent
\item[\hskip\labelsep
\scshape
#1]\ignorespaces
}{%
\endtrivlist\@endpefalse
}
\fi

\providecommand{\event}{DCM 2011} 
\usepackage{breakurl}
\def\titlerunning{Completeness of algebraic CPS simulations}
\def\authorrunning{A. Assaf and S.Perdrix}

\makeatother

\usepackage{babel}
\providecommand{\examplename}{Example}
\providecommand{\lemmaname}{Lemma}
\providecommand{\proofname}{Proof}
\providecommand{\propositionname}{Proposition}
\providecommand{\theoremname}{Theorem}

\begin{document}
\global\long\def\lambdalin{\lambda_{lin}}
\global\long\def\lambdalinr{\lambda_{lin}^{\step}}
\global\long\def\lambdaline{\lambda_{lin}^{=}}

\global\long\def\lambdaalg{\lambda_{alg}}
\global\long\def\lambdaalgr{\lambda_{alg}^{\step}}
\global\long\def\lambdaalge{\lambda_{alg}^{=}}

\global\long\def\step{\to}
\global\long\def\steps{\step^{*}}

\global\long\def\stepbv{\step_{\beta_{v}}}
\global\long\def\stepsbv{\step_{\beta_{v}}^{*}}
\global\long\def\stepbn{\step_{\beta_{n}}}
\global\long\def\stepsbn{\step_{\beta_{n}}^{*}}

\global\long\def\stepl{\step_{l}}
\global\long\def\steple{\step_{l}^{=}}
\global\long\def\stepsl{\step_{l}^{*}}
\global\long\def\stepsle{\step_{l}^{=*}}

\global\long\def\stepa{\step_{a}}
\global\long\def\stepae{\step_{a}^{=}}
\global\long\def\stepsa{\step_{a}^{*}}
\global\long\def\stepsae{\step_{a}^{=*}}

\global\long\def\steplb{\step_{l\cup\beta}}
\global\long\def\steplbe{\step_{l\cup\beta}^{=}}
\global\long\def\stepslb{\step_{l\cup\beta}^{*}}
\global\long\def\stepslbe{\step_{l\cup\beta}^{=*}}

\global\long\def\stepab{\step_{a\cup\beta}}
\global\long\def\stepabe{\step_{a\cup\beta}^{=}}
\global\long\def\stepsab{\step_{a\cup\beta}^{*}}
\global\long\def\stepsabe{\step_{a\cup\beta}^{=*}}

\global\long\def\translatel#1{[\![#1]\!]}

\global\long\def\translatea#1{\{\!|#1|\!\}}

\global\long\def\lam#1{\lambda#1.\,}

\global\long\def\cpy{\mathsf{copy}}

\global\long\def\id{\mathsf{id}}

\title{Completeness of algebraic CPS simulations}

\author{Ali Assaf \institute{LIG, Universit\'e Joseh Fourier\\
 Grenoble, France}
\institute{\'Ecole Polytechnique\\
 Palaiseau, France}
\email{Ali.Assaf@imag.fr} \and Simon Perdrix \institute{CNRS, LIG, Universit\'e de Grenoble\\
 Grenoble, France}
\email{Simon.Perdrix@imag.fr}}
\maketitle
\begin{abstract}
The \emph{algebraic lambda calculus} ($\lambdaalg$) and the \emph{linear
algebraic lambda calculus} ($\lambdalin$) are two extensions of the
classical lambda calculus with linear combinations of terms. They
arise independently in distinct contexts: the former is a fragment
of the differential lambda calculus, the latter is a candidate lambda
calculus for quantum computation. They differ in the handling of application
arguments and algebraic rules. The two languages can simulate each
other using an algebraic extension of the well-known call-by-value
and call-by-name CPS translations. These simulations are sound, in
that they preserve reductions. In this paper, we prove that the simulations
are actually complete, strengthening the connection between the two
languages.
\end{abstract}

\section{Introduction}

\paragraph{Algebraic lambda calculi}

The \emph{algebraic lambda calculus} ($\lambdaalg$) \cite{VauxMSCS09}
and the \emph{linear algebraic lambda calculus} ($\lambdalin$) \cite{ArrighiDowekRTA08}
are two languages that extend the classical lambda calculus with linear
combinations of terms such as $\alpha.M+\beta.N$. They have been
introduced independently in two different contexts. The former is
a fragment of the differential lambda calculus, and has been introduced
in the context of linear logic with the purpose of quantifying non-determinism:
each term of a linear combination represents a possible evolution
in a non deterministic setting. The latter has been introduced as
a candidate for a language of quantum computation, where a linear
combination of terms corresponds to a superposition of states such
as $\frac{1}{\sqrt{2}}.|0\rangle+\frac{1}{\sqrt{2}}.|1\rangle$. The strength of  $\lambdalin$ is to allow superpositions of any terms without distinguishing programs and data, whereas most of the candidate languages for quantum computation are based on the slogan \emph{quantum data, classical control} \cite{Selinger04towardsa,SelingerV06,Perdrix08}.

The two languages, $\lambdaalg$ and $\lambdalin$, differ in their operational semantics. It turns
out that the first follows a \emph{call-by-name} strategy while the
second follows the equivalent of a \emph{call-by-value} strategy.
For example, in $\lambdaalg$ the term $(\lam xfxx)(\alpha.y+\beta.z)$
reduces as follows:

\begin{eqnarray*}
(\lam xfxx)(\alpha.y+\beta.z) & \step & f(\alpha.y+\beta.z)(\alpha.y+\beta.z)
\end{eqnarray*}
However, this does not agree with the nature of quantum computing.
It leads to the cloning of the state $\alpha.y+\beta.z$, which contradicts
the \emph{no-cloning} theorem \cite{WoottersZurekNATURE82}. Only
copying of base terms such as $y$ is allowed. Therefore, $\lambdalin$
reduces the term as follows.

\begin{eqnarray*}
(\lam xfxx)(\alpha.y+\beta.z) & \step & (\lam xfxx)(\alpha.y)+(\lam xfxx)(\beta.z)\\
 & \step & \alpha.(\lam xfxx)y+\beta.(\lam xfxx)z\\
 & \step & \alpha.fyy+\beta.fzz
\end{eqnarray*}

Despite these differences, the work in \cite{DiazcaroPerdrixTassonValiron11}
showed that the two languages can simulate each other. This was accomplished
by defining a translation from one language to the other. Given a
term $M$ of $\lambdalin$, we can encode it into a term $N$ of $\lambdaalg$
such that reductions of $M$ in $\lambdalin$ correspond to reductions
of $N$ in $\lambdaalg$. The translation is an algebraic extension
of the classical \emph{continuation-passing style} (CPS) encoding
used for simulating call-by-name and call-by-value \cite{HatcliffDanvySPPL94,PlotkinTCS75,SabryWadlerTPLS96}.

\paragraph{Contribution}

The CPS transformations introduced in \cite{DiazcaroPerdrixTassonValiron11}
have been proven to be sound, \textit{i.e.} if a term $M$ reduces
to a value $V$ in the source language, then the translation of $M$
reduces to the translation of $V$ in the target language. In this
paper we prove that they are actually complete, \textit{i.e.} that
the converse is also true: if the translation of $M$ reduces to the
translation of $V$ in the target language, then $M$ reduces to $V$
in the source language. We do so by modifying techniques used by Sabry
and Wadler in \cite{SabryWadlerTPLS96} to define an inverse translation
and showing that it also preserves reductions. The completeness of
these CPS transformations strengthens the connection between works
done in linear logic \cite{EhrhardMSCS05,EhrhardLICS10,EhrhardRegnierTCS03,VauxRTA07}
and works on quantum computation \cite{AltenkirchGrattageLICS05,ArrighiDiazcaroQPL09,ArrighiDowekWRLA04,ValironDCM10}.

\paragraph{Plan of the paper}

The rest of the paper is structured as follows. In section \ref{sec:language},
the syntax and the reduction rules of both algebraic languages are
presented. Section \ref{sec:cbv-to-cbn} is dedicated to the simulation
of $\lambdalin$ by $\lambdaalg$, and section \ref{sec:cbn-to-cbv}
to the opposite simulation. In each of the two cases, the translation
introduced in \cite{DiazcaroPerdrixTassonValiron11} is presented,
the grammar of the encoded terms in the target language is given,
the inverse translation is defined, and finally the completeness of
the CPS translation is proven.

\section{\label{sec:language}The algebraic lambda calculi}

The languages $\lambdalin$ and $\lambdaalg$ share the same syntax,
defined by the following grammar, where $\alpha$ ranges over a defined
ring, the \emph{ring of scalars}.

\[
\begin{array}{rcll}
M,N,L & ::= & V\mid MN\mid\alpha.M\mid M+N & \mbox{(terms)}\\
U,V,W & ::= & B\mid0\mid\alpha.V\mid V+W & \mbox{(values)}\\
B & ::= & x\mid\lam xM & \mbox{(base values)}
\end{array}
\]

We can form sums of terms and multiplication by scalars, and there
is a neutral element $0$. The values we consider are formed by taking
linear combinations of base values, \textit{i.e.} variables and abstractions.
This gives the languages the structure of a vector space (a \emph{module}
to be precise).

We describe the operational semantics of the two languages using small-step
rewrite rules. The rules are presented in Figure \ref{fig:rewrite-rules}.
As mentioned, $\lambdaalg$ substitutes the argument directly in the
body of a function, while $\lambdalin$ delays the substitution until
the argument is a base value. We use the same notation as in \cite{DiazcaroPerdrixTassonValiron11}
to define the following rewrite systems obtained by combining the
rules described in Figure \ref{fig:rewrite-rules} .

\[
\begin{array}{cc}
\begin{array}{lcl}
\step_{\beta_{n}} & ::= & \beta_{n}\cup\xi\\
\stepa & ::= & A\cup L\cup\xi
\end{array} & \begin{array}{lcl}
\step_{\beta_{v}} & ::= & \beta_{v}\cup\xi\cup\xi_{\lambda_{lin}}\\
\stepl & ::= & A_{l}\cup A_{r}\cup L\cup\xi\cup\xi_{\lambda_{lin}}
\end{array}\end{array}
\]
The rewrite systems for the two languages are then defined as follows.

\begin{center}
\begin{tabular}{|c|c|}
\hline 
Language & Rewrite system\tabularnewline
\hline 
\hline 
$\lambdalin$ & $\steplb::=(\stepl)\cup(\step_{\beta_{v}})$\tabularnewline
\hline 
$\lambdaalg$ & $\stepab::=(\stepa)\cup(\step_{\beta_{n}})$\tabularnewline
\hline 
\end{tabular}
\par\end{center}

\begin{figure}
\[
\begin{array}{cc}
\hline \multicolumn{2}{c}{\mbox{Rules specific to \ensuremath{\lambdaalg}}}\\
\hline \\
\mbox{Call-by-name (\ensuremath{\beta_{n}})} & \mbox{Linearity of application (\ensuremath{A})}\\[2ex]
\begin{array}{rcl}
(\lam xM)N & \to & M[x:=N]\end{array} & \begin{array}{rcl}
(M+N)L & \to & ML+NL\\
(\alpha.M)N & \to & \alpha.(MN)\\
(0)M & \to & 0
\end{array}\\
\\
\hline \multicolumn{2}{c}{\mbox{Rules specific to \ensuremath{\lambdalin}}}\\
\hline \\
\mbox{Call-by-value (\ensuremath{\beta_{v}})} & \mbox{Right context rule (\ensuremath{\xi_{\lambdalin}})}\\[2ex]
\begin{array}{rcl}
(\lam xM)B & \to & M[x:=B]\end{array} & \dfrac{M\to M'}{VM\to VM'}\\[2ex]
\\
\mbox{Left linearity of application (\ensuremath{A_{l}})} & \mbox{Right linearity of application (\ensuremath{A_{r}})}\\[2ex]
\begin{array}{rcl}
(M+N)V & \to & MV+NV\\
(\alpha.M)V & \to & \alpha.(MV)\\
(0)V & \to & 0
\end{array} & \begin{array}{rcl}
B(M+N) & \to & BM+BN\\
B(\alpha.M) & \to & \alpha.(BM)\\
B(0) & \to & 0
\end{array}\\
\\
\hline \multicolumn{2}{c}{\mbox{Common rules}}\\
\hline \\
\multicolumn{2}{c}{\mbox{Vector space rules (\ensuremath{L={\it Asso}\cup{\it Com}\cup F\cup S})}}\\[2ex]
\mbox{Associativity (\ensuremath{{\it Asso}})} & \mbox{Commutativity (\ensuremath{{\it Com}})}\\[2ex]
\begin{array}{rcl}
M+(N+L) & \to & (M+N)+L\\
(M+N)+L & \to & M+(N+L)
\end{array} & \begin{array}{rcl}
M+N & \to & N+M\end{array}\\[2ex]
\\
\mbox{Factorization (\ensuremath{F})} & \mbox{Simplification (\ensuremath{S})}\\[2ex]
\begin{array}{rcl}
\alpha.M+\beta.M & \to & (\alpha+\beta).M\\
\alpha.M+M & \to & (\alpha+1).M\\
M+M & \to & (1+1).M\\
\alpha.(\beta.M) & \to & (\alpha\beta).M
\end{array} & \begin{array}{rcl}
\alpha.(M+N) & \to & \alpha.M+\alpha.N\\
1.M & \to & M\\
0.M & \to & 0\\
\alpha.0 & \to & 0\\
0+M & \to & M
\end{array}\\[2ex]
\\
\multicolumn{2}{c}{\mbox{Context rules (\ensuremath{\xi})}}\\[2ex]
\dfrac{M\to M'}{(M)~N\to(M')~N} & \dfrac{M\to M'}{M+N\to M'+N}\\[2ex]
\dfrac{M\to M'}{\alpha.M\to\alpha.M'} & \dfrac{N\to N'}{M+N\to M+N'}\\
\\
\hline \end{array}
\]

\caption{\label{fig:rewrite-rules}Rewrite rules for $\lambdalin$ and $\lambdaalg$}
\end{figure}

\begin{example}
\label{ex:cbn-vs-cbv}Let $\langle M,N\rangle:=\lam ffMN$ be the
Church encoding of pairs, let $\cpy=\lam x\langle x,x\rangle$, and
consider the term $\cpy(y+z)$. The term reduces in $\lambdaalg$:
\begin{eqnarray*}
\cpy(y+z) & = & (\lam x\langle x,x\rangle)(y+z)\\
 & \stepbn & \langle y+z,y+z\rangle
\end{eqnarray*}
 As mentioned above, the term $y+z$ is cloned, and if it represented
quantum superposition this would violate the no-cloning theorem. In
$\lambdalin$, the term reduces instead as:
\end{example}
\begin{eqnarray*}
\cpy(y+z) & = & (\lam x\langle x,x\rangle)(y+z)\\
 & \stepl & (\lam x\langle x,x\rangle)y+(\lam x\langle x,x\rangle)z\\
 & \stepbv & \langle y,y\rangle+(\lam x\langle x,x\rangle)z\\
 & \stepbv & \langle y,y\rangle+\langle z,z\rangle
\end{eqnarray*}

\section{\label{sec:cbv-to-cbn}Completeness of the call-by-value to call-by-name
simulation}

The translation in \cite{DiazcaroPerdrixTassonValiron11} is a direct
extension of the classical CPS encoding used by Plotkin \cite{PlotkinTCS75}
to show that the call-by-name lambda calculus simulates call-by-value.
The definition is the following.

\begin{eqnarray*}
\translatel x & = & \lambda k.kx\\
\translatel{\lambda x.M} & = & \lambda k.k(\lambda x.\translatel M)\\
\translatel{MN} & = & \lambda k.\translatel M(\lambda b_{1}.\translatel N(\lambda b_{2}.b_{1}b_{2}k))\\
\translatel 0 & = & 0\\
\translatel{\alpha.M} & = & \lambda k.(\alpha.\translatel M)k\\
\translatel{M+N} & = & \lambda k.(\translatel M+\translatel N)k
\end{eqnarray*}

This translation simulates the reductions of a term $M$ by the reductions
of the term $\translatel Mk$, where $k$ is free. It works the same
way as the classical CPS simulation: instead of returning the result
of a computation, all terms receive an additional argument $k$ called
the \emph{continuation}, which describes the rest of the computation.
This technique makes evaluation order, intermediate values, and function
returns explicit, which allows us to encode the proper evaluation
strategy.

The translation preserves the set of free variables. New variables
names like $k$, $b$, $b_{1}$ or $b_{2}$ are chosen to be fresh
so as to not collide with free variables in the term. We reserve and
always use the name $k$ to abstract over continuations, and the names
$b$, $b_{1}$, and $b_{2}$ for intermediate values. It is a fact
that these variables never clash with each other.
\begin{example}
The reductions of the term $\cpy(y+z)$ in $\lambdalin$ are simulated
in $\lambdaalg$ by the following reductions:

\begin{eqnarray*}
\translatel{\cpy(y+z)}k & = & \left(\lam k\translatel{\cpy}\left(\lam{b_{1}}\translatel{y+z}\left(\lam{b_{2}}b_{1}b_{2}k\right)\right)\right)k\\
 & \step_{\beta_{n}} & \translatel{\cpy}\left(\lam{b_{1}}\translatel{y+z}\left(\lam{b_{2}}b_{1}b_{2}k\right)\right)\\
 & \stepbn & \left(\lam{b_{1}}\translatel{y+z}\left(\lam{b_{2}}b_{1}b_{2}k\right)\right)\left(\lam x\translatel{\langle x,x\rangle}\right)\\
 & \stepbn & \translatel{y+z}\left(\lam{b_{2}}\left(\lam x\translatel{\langle x,x\rangle}\right)b_{2}k\right)\\
 & \stepbn & \left(\translatel y+\translatel z\right)\left(\lam{b_{2}}\left(\lam x\translatel{\langle x,x\rangle}\right)b_{2}k\right)\\
 & \stepa & \translatel y\left(\lam{b_{2}}\left(\lam x\translatel{\langle x,x\rangle}\right)b_{2}k\right)+\translatel z\left(\lam{b_{2}}\left(\lam x\translatel{\langle x,x\rangle}\right)b_{2}k\right)\\
 & \stepsab & \left(\lam{b_{2}}\left(\lam x\translatel{\langle x,x\rangle}\right)b_{2}k\right)y+\left(\lam{b_{2}}\left(\lam x\translatel{\langle x,x\rangle}\right)b_{2}k\right)z\\
 & \stepsab & \left(\lam x\translatel{\langle x,x\rangle}\right)yk+\left(\lam x\translatel{\langle x,x\rangle}\right)zk\\
 & \stepsab & \translatel{\langle y,y\rangle}k+\translatel{\langle z,z\rangle}k
\end{eqnarray*}

We see that the result is the one that corresponds to call-by-value.
As expected, there was no cloning.
\end{example}
Notice in the example above that there are many more steps in the
simulation than in the original reduction sequence in Example \ref{ex:cbn-vs-cbv}.
A lot of the steps replace the continuation variables and intermediate
variables introduced by the translation. In a sense, all these intermediary
terms represent the {}``same'' term in the source language, and
we call these intermediary steps \emph{administrative reductions}.

To deal with this, we use an intermediate translation denoted by $M:K$.
This \emph{colon }translation was originally used by Plotkin \cite{PlotkinTCS75}
to describe intermediate reductions of translated terms, where initial
\emph{administrative redexes }had been eliminated.

\[
\begin{array}{cc}
\begin{array}{rcl}
\Psi(x) & = & x\\
\Psi(\lam xM) & = & \lam x\translatel M\\
B:K & = & K\Psi(B)\\
0:K & = & 0\\
\alpha.M:K & = & \alpha.(M:K)\\
M+N:K & = & M:K+N:K
\end{array} & \begin{array}{ccc}
BN:K & = & N:\lam b\Psi(B)bK\\
(MN)L:K & = & MN:\lam{b_{1}}\translatel L(\lam{b_{2}}b_{1}b_{2}K)\\
(0)N:K & = & 0:K\\
(\alpha.M)N:K & = & \alpha.(MN):K\\
(M+N)L:K & = & ML+NL:K
\end{array}\end{array}
\]

This CPS translation was proved to be sound by showing that it preserves
reductions: for any term $M$, if $M$ reduces to $M'$, then $M:K$
reduces to $M':K$ for all $K$. Combined with the fact that $\translatel Mk$
reduces initially to $M:k$, this gave the soundness of the simulation.
\begin{prop}[Soundness \cite{DiazcaroPerdrixTassonValiron11}]
\label{prop:Soundness} For any term $M$, if $M\stepslb V$ then
$\translatel Mk\stepsab V:k$.
\end{prop}
The goal of this paper is to show that the converse is also true:
\begin{thm}[Completeness]
\label{thm:completeness} If $\translatel Mk\stepsab V:k$ then $M\stepslb V$.
\end{thm}
To prove it, we define an inverse translation and show that it preserves
reductions. First, we need to characterize the\emph{ }structure of
the encoded terms. We define a subset of $\lambdaalg$ which contains
the image of the translation and is closed by $\stepab$ reductions
with the following grammar:

\[
\begin{array}{rcll}
C & ::= & KB\mid B_{1}B_{2}K\mid TK & \mbox{(base computations)}\\
D & ::= & C\mid0\mid\alpha.D\mid D_{1}+D_{2} & \mbox{(computation combinations)}\\
\\
S & ::= & \lam kC & \mbox{(base suspensions)}\\
T & ::= & S\mid0\mid\alpha.T\mid T_{1}+T_{2} & \mbox{(suspension combinations)}\\
\\
K & ::= & k\mid\lam bBbK\mid\lam{b_{1}}T(\lam{b_{2}}b_{1}b_{2}K) & \mbox{(continuations)}\\
\\
B & ::= & x\mid\lam xS & \mbox{(CPS-values)}
\end{array}
\]

There are four main categories of terms: \emph{computations}, \emph{suspensions},
\emph{continuations}, and \emph{CPS-values}. We distinguish base computations
$C$ from linear combinations of computations $D$, as well as base
suspensions $S$ from linear combinations of suspensions $T$. The
translation $\translatel M$ gives a term of the class $T$, while
$\translatel Mk$ and $M:K$ are of class $D$. One can easily check
that each of the classes $D$, $T$, $K$ and $B$ is closed by $\stepab$
reductions.

There are some restrictions on the names of the variables in this
grammar. The variable name $k$ that appears in the class $K$ must
be the same as the one used in suspensions of the form $\lam kC$.
It cannot appear as a variable name in any other term. This is to
agree with the requirement of freshness that we mentioned above. The
same applies for the variables $b$, $b_{1}$ and $b_{2}$: they cannot
appear (free) in any sub-term. In particular, these restrictions ensure
that the grammar for each category is unambiguous. The three kinds
of variables ($x$, $k$ and $b$) play different roles, which is
why we distinguish them using different names.

Computations are the terms that simulate the steps of the reductions,
hence the name. They are the only terms that contain applications,
so they are the only terms that can $\beta$-reduce. In fact, notice
that the arguments in applications are always base values. This shows
a simple alternative proof for the \emph{indifference }property \cite{DiazcaroPerdrixTassonValiron11}
of the CPS translation, namely that the reductions of a translated
term are exactly the same in $\lambdalin$ and $\lambdaalg$.
\begin{prop}[Indifference \cite{DiazcaroPerdrixTassonValiron11}]
 For any computations $D$ and $D'$, $D\stepab D'$ if and only
if $D\steplb D'$. In particular, if $M\stepslb V$ then $\translatel Mk\stepslb V:k$.
\end{prop}
We define the inverse translation using the following four functions,
corresponding to each of the four main categories in the grammar.

\[
\begin{array}{rclrcl}
\overline{KB} & = & \underline{K}[\psi(B)] & \sigma(\lam kC) & = & \overline{C}\\
\overline{B_{1}B_{2}K} & = & \underline{K}[\psi(B_{1})\psi(B_{2})] & \sigma(0) & = & 0\\
\overline{TK} & = & \underline{K}[\sigma(T)] & \sigma(\alpha.T) & = & \alpha.\sigma(T)\\
\overline{0} & = & 0 & \sigma(T_{1}+T_{2}) & = & \sigma(T_{1})+\sigma(T_{2})\\
\overline{\alpha.D} & = & \alpha.\overline{D}\\
\overline{D_{1}+D_{2}} & = & \overline{D_{1}}+\overline{D_{2}}\\
 &  &  & \underline{k}[M] & = & M\\
\psi(x) & = & x & \underline{\lam bBbK}[M] & = & \underline{K}[\psi(B)M]\\
\psi(\lam xS) & = & \lam x\sigma(S) & \underline{\lam{b_{1}}T(\lam{b_{2}}b_{1}b_{2}K)}[M] & = & \underline{K}[M\sigma(T)]
\end{array}
\]

These functions are well-defined because the grammar for each category
is unambiguous. To prove the completeness of the simulation we need
a couple of lemmas. The first two state that the translation defined
above is in fact an inverse.
\begin{lem}
\label{lem:inverse-term} For any term $M$, $\overline{\translatel Mk}=M$.\end{lem}
\begin{proof}
We have $\overline{\translatel Mk}=\underline{k}[\sigma(\translatel M)]=\sigma(\translatel M)$
so we have to show that $\sigma(\translatel M)=M$ for all $M$. The
proof follows by induction on the structure of $M$.
\end{proof}
In general, $\overline{M:k}\neq M$. Although it would be true for
a classical translation, it does not hold in the algebraic case. Specifically,
we have $(\alpha.M)L:k=\alpha.(ML):k$ and $(M+N)L:k=ML+NL:K$, so
the translation is not injective. However it is still true for values.
\begin{lem}
\label{lem:inverse-value} For any value $V$, $\overline{V:k}=V$.\end{lem}
\begin{proof}
By induction on the structure of $V$.
\end{proof}
The third lemma that we need states that the inverse translation preserves
reductions.
\begin{lem}
\label{lem:inverse-step} For any computation $D$, if $D\stepab D'$
then $\overline{D}\stepslb\overline{D'}$.
\end{lem}
With these we can prove the completeness theorem.
\begin{proof}[Proof of Theorem \ref{thm:completeness}]
 By using Lemma \ref{lem:inverse-step} for each step of the reduction,
we get $\overline{\translatel Mk}\stepslb\overline{V:k}$. By Lemma
\ref{lem:inverse-term} and Lemma \ref{lem:inverse-value}, this implies
$M\stepslb V$.
\end{proof}
To prove Lemma \ref{lem:inverse-step}. we need several intermediary
lemmas.
\begin{lem}[Substitution]
\label{lem:substitution-lemma} The following are true.
\begin{enumerate}
\item \textup{$\psi(B_{1})[x:=\psi(B)]=\psi(B_{1}[x:=B])$}
\item $\sigma(T)[x:=\psi(B)]=\sigma(T[x:=B])$
\item \textup{$\overline{C}[x:=\psi(B)]=\overline{C[x:=B]}$}
\item $\underline{K}[M][x:=\psi(B)]=\underline{K[x:=B]}[M[x:=\psi(B)]]$
\end{enumerate}
\end{lem}
\begin{proof}
By induction on the structure of $B_{1}$, $T$, $C$, and $K$.
\end{proof}
The next lemma states that we can compose two continuations $K_{1}$
and $K_{2}$ by replacing $k$ by $K_{1}$ in $K_{2}$.
\begin{lem}
\label{lem:.continuation-composition}For all terms $M$ and continuations
$K_{1}$ and $K_{2}$, $\underline{K_{1}}[\underline{K_{2}}[M]]=\underline{K_{2}[k:=K_{1}]}[M]$.\end{lem}
\begin{proof}
By induction on the structure of $K_{2}$.\end{proof}
\begin{lem}
\textup{\label{lem:continuation-substitution}For all $K$ and $C$,
$\underline{K}[\overline{C}]=\overline{C[k:=K]}$.}\end{lem}
\begin{proof}
By induction on the structure of $C$, using Lemma \ref{lem:.continuation-composition}
where necessary.
\end{proof}
The following lemma is essential to the preservation of reductions.
It shows that reductions of a term $M$ can always be carried in the
context $\underline{K}[M]$.
\begin{lem}
\label{lem:continuation-step} For any continuation $K$ and term
$M$, if $M\steplb M'$, then $\underline{K}[M]\steplb\underline{K}[M']$.\end{lem}
\begin{proof}
By induction on the structure of $K$.\end{proof}
\begin{lem}
\label{lem:continuation-linearity}The following are true.
\begin{itemize}
\item $\underline{K}[M_{1}+M_{2}]\stepsl\underline{K}[M_{1}]+\underline{K}[M_{2}]$
\item $\underline{K}[\alpha.M]\stepsl\alpha.\underline{K}[M]$
\item $\underline{K}[0]\stepsl0$
\end{itemize}
\end{lem}
\begin{proof}
We prove each statement by induction on $K$, using Lemma \ref{lem:continuation-step}
where necessary.\end{proof}
\begin{lem}
\label{lem:suspension-step} For any suspension $T$, if $T\stepa T'$
then $\sigma(T)\stepl\sigma(T')$.\end{lem}
\begin{proof}
By induction on the reduction rule. Since $T$ terms do not contain
applications, the only cases possible are $L\cup\xi$, which are common
to both languages.
\end{proof}
We now have the tools to finish the proof of \ref{lem:inverse-step}.
\begin{proof}[Proof of Lemma \ref{lem:inverse-step}]
 By induction on the reduction rule, using Lemmas \ref{lem:substitution-lemma},
\ref{lem:continuation-substitution}, \ref{lem:continuation-step},
\ref{lem:continuation-linearity} and \ref{lem:suspension-step} where
necessary
\end{proof}

\section{\label{sec:cbn-to-cbv}Completeness of the call-by-name to call-by-value
simulation}

The simulation in this direction is similar to the other one, and
uses the same techniques. The adjustments we have to make are the
same as in the classical case, and deal mainly with our treatment
of variables and applications. The CPS translation, as defined in
\cite{DiazcaroPerdrixTassonValiron11}, is the following.

\begin{eqnarray*}
\translatea x & = & x\\
\translatea{\lambda x.M} & = & \lambda k.k(\lambda x.\translatea M)\\
\translatea{MN} & = & \lambda k.\translatea M(\lambda b.b\translatea Nk)\\
\translatea 0 & = & 0\\
\translatea{\alpha.M} & = & \lambda k.(\alpha.\translatea M)k\\
\translatea{M+N} & = & \lambda k.(\translatea M+\translatea N)k
\end{eqnarray*}

Again, this translation simulates the reductions of a term $M$ by
the reductions of the term $\translatea Mk$, where $k$ is free.
\begin{example}
The reductions of the term $\cpy(y+z)$ in $\lambdaalg$ are simulated
in $\lambdalin$ by the following reductions.

\begin{eqnarray*}
\translatea{\cpy(y+z)}k & = & \left(\lam k\translatea{\cpy}\left(\lam bb\translatea{y+z}k\right)\right)k\\
 & \stepbv & \translatea{\cpy}\left(\lam bb\translatea{y+z}k\right)\\
 & \stepbv & \left(\lam bb\translatea{y+z}k\right)\left(\lam x\translatea{\langle x,x\rangle}\right)\\
 & \stepbv & \left(\lam x\translatea{\langle x,x\rangle}\right)\translatea{y+z}k\\
 & \stepslb & \translatea{\langle y+z,y+z\rangle}k
\end{eqnarray*}

We see that the result is the one that corresponds to call-by-name.
It is natural to ask how we were able to perform this cloning of the
state $y+z$ in a call-by-value setting and how it can agree with
the no-cloning theorem. The answer is that the CPS encoding of the
term $y+z$ is $\translatea{y+z}=\lam k(x+y)k$, which is an abstraction.
In the quantum point of view, we can interpret this as a program,
or a specification, that \emph{prepares }the quantum state $x+y$.
Therefore this program can be duplicated.
\end{example}
The soundness of the simulations uses a similar colon translation.

\[
\begin{array}{cc}
\begin{array}{rcl}
\Phi(\lam xM) & = & \lam x\translatea M\\
\lam xM:K & = & K\Phi(\lam xM)\\
x:K & = & xK\\
0:K & = & 0\\
\alpha.M:K & = & \alpha.(M:K)\\
M+N:K & = & M:K+N:K
\end{array} & \begin{array}{ccc}
(\lam xM)N:K & = & \Phi(\lam xM)\translatea NK\\
xN:K & = & x:(\lam bb\translatea NK)\\
(MN)L:K & = & MN:\lam bb\translatea LK\\
(0)N:K & = & 0:K\\
(\alpha.M)N:K & = & \alpha.(MN):K\\
(M+N)L:K & = & ML+NL:K
\end{array}\end{array}
\]

\begin{prop}[Soundness \cite{DiazcaroPerdrixTassonValiron11}]
 For any term $M$, if $M\stepsab V$ then $\translatea Mk\stepslb V:k$.
\end{prop}
We will use the same procedure as in the previous section to show
that the translation is also complete.
\begin{thm}[Completeness]
\label{thm:completeness-a} If $\translatea Mk\stepslb V:k$ then
$M\stepsab V$.
\end{thm}
Here is the grammar of the target language. It is closed under $\steplb$
reductions.

\[
\begin{array}{rcll}
C & ::= & KB\mid BSK\mid TK & \mbox{(base computations)}\\
D & ::= & C\mid0\mid\alpha.D\mid D_{1}+D_{2} & \mbox{(computation combinations)}\\
S & ::= & x\mid\lam kC & \mbox{(base suspensions)}\\
T & ::= & S\mid0\mid\alpha.T\mid T_{1}+T_{2} & \mbox{(suspension combinations)}\\
K & ::= & k\mid\lam bbSK & \mbox{(continuations)}\\
B & ::= & \lam xS & \mbox{(CPS-values)}
\end{array}
\]

Notice how $x$ is now considered a suspension, not a CPS-value. This
is because $x$ is replaced by a suspension after beta-reducing a
term of the form $(\lam xS)SK$. This is the main difference between
the call-by-name and call-by-value CPS simulations. Other than that,
it satisfies the same properties. In particular, we have the same
indifference property.
\begin{prop}[Indifference \cite{DiazcaroPerdrixTassonValiron11}]
 For any computations $D$ and $D'$, $D\stepab D'$ if and only
if $D\steplb D'$. In particular, if $M\stepsab V$ then $\translatea Mk\stepsab V:k$.
\end{prop}
We define the inverse translation using the following four functions.

\[
\begin{array}{rclrcl}
\overline{KB} & = & \underline{K}[\phi(B)] & \sigma(x) & = & x\\
\overline{BSK} & = & \underline{K}[\phi(B)\sigma(S)] & \sigma(\lam kC) & = & \overline{C}\\
\overline{TK} & = & \underline{K}[\sigma(T)] & \sigma(0) & = & 0\\
\overline{0} & = & 0 & \sigma(\alpha.T) & = & \alpha.\sigma(T)\\
\overline{\alpha.D} & = & \alpha.\overline{D} & \sigma(T_{1}+T_{2}) & = & \sigma(T_{1})+\sigma(T_{2})\\
\overline{D_{1}+D_{2}} & = & \overline{D_{1}}+\overline{D_{2}}\\
 &  &  & \underline{k}[M] & = & M\\
\phi(\lam xS) & = & \lam x\sigma(S) & \underline{\lam bbSK}[M] & = & \underline{K}[M\sigma(S)]
\end{array}
\]

To prove the completeness of the simulation we need analogous lemmas.
Their proofs are similar, but we need to account for the changes mentioned
above.
\begin{lem}
\label{lem:inverse-term-a} For any term $M$,$\overline{\translatea Mk}=M$.\end{lem}
\begin{proof}
We have $\overline{\translatea Mk}=\underline{k}[\sigma(\translatea M)]=\sigma(\translatea M)$
so we have to show that $\sigma(\translatea M)=M$ for all $M$. The
proof follows by induction on the structure of $M$.\end{proof}
\begin{lem}
\label{lem:inverse-value-a} For any value $V$, $\overline{V:k}=V$.\end{lem}
\begin{proof}
By induction on the structure of $V$.\end{proof}
\begin{lem}
\label{lem:inverse-step-a} For any computation $D$, if $D\steplb D'$
then $\overline{D}\stepsab\overline{D'}$.
\end{lem}
With these we can prove the completeness theorem.
\begin{proof}[Proof of Theorem \ref{thm:completeness-a}]
 By using Lemma \ref{lem:inverse-step-a} for each step of the reduction,
we get $\overline{\translatea Mk}\stepsab\overline{V:k}$. By Lemma
\ref{lem:inverse-term-a} and Lemma \ref{lem:inverse-value-a}, this
implies $M\stepsab V$.
\end{proof}
To prove Lemma \ref{lem:inverse-step-a}, we need similar intermediary
lemmas.
\begin{lem}[Substitution]
\label{lem:substitution-lemma-a} The following are true.
\begin{enumerate}
\item \textup{$\phi(B)[x:=\sigma(S)]=\phi(B[x:=S])$}
\item $\sigma(T)[x:=\sigma(S)]=\sigma(T[x:=S])$
\item \textup{$\overline{C}[x:=\sigma(S)]=\overline{C[x:=S]}$}
\item $\underline{K}[M][x:=\sigma(S)]=\underline{K[x:=S]}[M[x:=\sigma(S)]]$
\end{enumerate}
\end{lem}
\begin{proof}
By induction on the structure of $B$, $T$, $C$ and $K$.\end{proof}
\begin{lem}
\label{lem:.continuation-composition-a} For all terms $M$ and continuations
$K_{1}$ and $K_{2}$, $\underline{K_{1}}[\underline{K_{2}}[M]]=\underline{K_{2}[k:=K_{1}]}[M]$.\end{lem}
\begin{proof}
By induction on the structure of $K_{2}$.\end{proof}
\begin{lem}
\textup{\label{lem:continuation-substitution-a} For all $K$ and
$C$, $\underline{K}[\overline{C}]=\overline{C[k:=K]}$}.\end{lem}
\begin{proof}
By induction on the structure of $C$, using Lemma \ref{lem:.continuation-composition-a}
where necessary.\end{proof}
\begin{lem}
\label{lem:continuation-step-a}For any continuation $K$ and term
$M$, if $M\stepab M'$ then $\underline{K}[M]\stepab\underline{K}[M']$.\end{lem}
\begin{proof}
By induction on the structure of $K$.\end{proof}
\begin{lem}
\label{lem:continuation-linearity-a}The following are true.
\begin{itemize}
\item $\underline{K}[M_{1}+M_{2}]\stepsa\underline{K}[M_{1}]+\underline{K}[M_{2}]$
\item $\underline{K}[\alpha.M]\stepsa\alpha.\underline{K}[M]$
\item $\underline{K}[0]\stepsa0$
\end{itemize}
\end{lem}
\begin{proof}
We prove each statement by induction on $K$, using Lemma \ref{lem:continuation-step-a}
where necessary.\end{proof}
\begin{lem}
\label{lem:suspension-step-a} For any suspension $T$, if $T\stepl T'$
then $\sigma(T)\stepa\sigma(T')$.\end{lem}
\begin{proof}
By induction on the reduction rule. Since $T$ terms do not contain
applications, the only cases possible are $L\cup\xi$, which are common
to both languages.
\end{proof}
We can now prove Lemma \ref{lem:inverse-step-a}.
\begin{proof}[Proof of Lemma \ref{lem:inverse-step-a}]

By induction on the reduction rule, using Lemmas \ref{lem:substitution-lemma-a},
\ref{lem:continuation-substitution-a}, \ref{lem:continuation-step-a},
\ref{lem:continuation-linearity-a} and \ref{lem:suspension-step-a}
where necessary. Notice that the rules $\xi_{\lambda_{lin}}$ and
$A_{r}$ are not applicable since arguments in the target language
are always base terms.
\end{proof}

\section{\label{sec:conclusion}Discussion and conclusion}

We showed the completeness of two CPS translations simulating algebraic
lambda calculi introduced in \cite{DiazcaroPerdrixTassonValiron11}.
We did so by using techniques inspired from \cite{SabryWadlerTPLS96}
to define an inverse translation and showing that it preserves reductions.

Our treatment differs from Sabry and Wadler's \cite{SabryWadlerTPLS96},
not only because they work in a non-algebraic setting, but also because
they decompile continuations into abstractions. For example, they
defined $\underline{\lam bBbk}[M]$ as $\mathsf{let}\ b=M\ \mathsf{in}\ \mbox{\ensuremath{\phi}(B)}b$.
This required the modification of the source language and led to the
consideration of the \emph{computational lambda calculus} \cite{MoggiLICS89}
as a source language instead. We avoid this by directly substituting
and eliminating variables introduced by the forward translation, which
allows us to obtain an exact inverse.

However, the translations defined in \cite{SabryWadlerTPLS96} satisfy
an additional property: they form a Galois connection. Our translations
fail to satisfy one of the four required criteria to be a Galois connection,
namely that $\overline{N}:k$ reduces to $N$. It would be interesting
to see if we can accomplish the same thing in the algebraic case,
all while dealing with the problems mentioned above.

Originally, the work in \cite{DiazcaroPerdrixTassonValiron11} also
considers another version of $\lambdalin$ and $\lambdaalg$ with
\emph{algebraic equalities} instead of \emph{algebraic reductions}.
For example, we could go back and forth between $M+N-N$ and $M$,
which is not permitted by the rules we presented above.\emph{ }Algebraic
equalities can be formulated as the symmetric closure of the algebraic
reductions $\stepa$ and $\stepl$. The resulting four systems $\lambdalinr$,
$\lambdaalgr$, $\lambdaline$, and $\lambdaalge$ have all been shown
to simulate each other. The results of this paper can be extended
to these systems as well.

\paragraph{Acknowledgments}

Many thanks to Alejandro Díaz-Caro, Benoît Valiron, Pablo Arrighi,
and Christophe Calvès for fruitful discussions and suggestions. This
work is supported by the CNRS - INS2I PEPS project QuAND.

\bibliographystyle{eptcs}
\bibliography{biblio}

\end{document}